\documentclass[11pt]{article}
\usepackage{fullpage}
\usepackage{amsmath,amsthm,amssymb}
\usepackage{hyperref}
\usepackage{thm-restate}
\usepackage[nameinlink,capitalise]{cleveref}
\usepackage{tikz}
\usepackage{xcolor}
\usepackage{url}
\usepackage{dsfont}

\makeatletter
\def\@fnsymbol#1{\ensuremath{\ifcase#1\or \dagger\or \ddagger\or
   \mathsection\or \mathparagraph\or \|\or **\or \dagger\dagger
   \or \ddagger\ddagger \else\@ctrerr\fi}}
\makeatother
\newcommand*\samethanks[1][\value{footnote}]{\footnotemark[#1]}

\DeclareMathOperator*{\E}{\mathrm{E}}
\DeclareMathOperator*{\argmax}{\mathrm{argmax}}
\DeclareMathOperator*{\argmin}{\mathrm{argmin}}

\newcommand{\mathsc}[1]{{\normalfont\textsc{#1}}}

\newcommand{\alg}{\mathsc{Alg}}
\newcommand{\reg}{\mathsf{Reg}}

\newcommand{\opt}{\mathsc{Opt}}
\newcommand{\ben}{\mathsf{OPT}}
\newcommand{\off}{\mathsc{Offset}}
\newcommand{\simple}{\mathsc{Greedy}}
\newcommand{\ep}{\varepsilon}

\newtheorem{theorem}{Theorem}[section]
\newtheorem{lemma}[theorem]{Lemma}
\newtheorem{corollary}[theorem]{Corollary}

\newtheorem*{theorem*}{Theorem}

\renewcommand{\d}{\ \mathrm{d}}

\renewcommand{\vec}[1]{{\mathbf{#1}}}

\definecolor{myblue}{rgb}{0.15, 0.1, 0.95}
\definecolor{mygreen}{rgb}{0.15, 0.65, 0.25}
\definecolor{myred}{rgb}{0.75, 0.25, 0.15}

\hypersetup{
    colorlinks=true,
    linkcolor=myred,
    citecolor = myblue,      
}

\title{Regret Minimization with Noisy Observations}

\author{Mohammad Mahdian\thanks{Google Research. Email: \texttt{\{mahdian,maojm\}@google.com}.} \and Jieming Mao\samethanks[1] \and Kangning Wang\thanks{Duke University. Email: \texttt{knwang@cs.duke.edu}.}}
\date{}

\begin{document}

\maketitle

\begin{abstract}
In a typical optimization problem, the task is to pick one of a number of options with the lowest cost or the highest value. In practice, these cost/value quantities often come through processes such as measurement or machine learning, which are noisy, with quantifiable noise distributions. To take these noise distributions into account, one approach is to assume a prior for the values, use it to build a posterior, and then apply standard stochastic optimization to pick a solution. However, in many practical applications, such prior distributions may not be available. In this paper, we study such scenarios using a regret minimization model.

In our model, the task is to pick the highest one out of $n$ values. The values are unknown and chosen by an adversary, but can be observed through noisy channels, where additive noises are stochastically drawn from known distributions. The goal is to minimize the \emph{regret} of our selection, defined as the expected difference between the highest and the selected value on the worst-case choices of values. We show that the na\"ive algorithm of picking the highest observed value has regret arbitrarily worse than the optimum, even when $n = 2$ and the noises are unbiased in expectation. On the other hand, we propose an algorithm which gives a constant-approximation to the optimal regret for any $n$. Our algorithm is conceptually simple, computationally efficient, and requires only minimal knowledge of the noise distributions.
\end{abstract}



\section{Introduction}
Classical optimization often deals with the problem of selecting an object from a family that maximizes a value or minimizes a cost. The implicit assumption is that the cost or value are given quantities, with no uncertainty around them. In practice, however, this assumption is often not accurate. Values are often measured (e.g, through physical measurements) or estimated (e.g., using machine learning algorithms) and we only observe noisy estimates of the value. To tackle uncertainty, a number of modifications to the classical paradigm are proposed. One extreme is stochastic optimization, which assumes values are drawn from known distributions. For example, if we have a prior on the values, and the measurement/estimation method has a known error distribution, one can use the Bayes rule to compute the posterior distribution of the values. The other extreme is robust optimization, which assumes no knowledge of the distribution of the values, and tries to find a solution that works well for the worst-case values in given ranges. This paper studies a model that falls between these two extremes: what if no prior on the values is known, but the measurement/estimation method has known noise distribution?  We believe this is a reasonable question in practice, since the noise distribution of particular measurement/estimation methods are often well-understood, but prior distributions depend on the specific domain, and might be hard to estimate. 

We model this question using the following prototypical model of selecting one item from a given set of $n$ items. Item $i$ has a value of $v_i$, and our goal is to pick the item with the highest $v_i$. The algorithm does not know $v_i$, but observes a noisy estimate $v_i+a_i$ of this value, where the noise $a_i$ is drawn (independently) from a known distribution $A_i$. To evaluate the performance of the algorithm, we adopt the standard framework of \emph{regret minimization}. Regret is defined as the worst-case difference between the optimal solution of an omniscient scheme (that observes the $v_i$ values) and that of an algorithm with limited information (only knowing the noise distributions and the noisy observations), averaged over the noises. 

We believe this model is practically appealing: dividing uncertainty into two categories, the uncertainty around the $v_i$ values and the uncertainty around the noises, and treating these two different kinds of uncertainty separately is well-motivated. In machine learning, two different types of uncertainty are distinguished~\cite{MLuncertainty}: {\em aleatoric} uncertainty, which is the noise inherent in the setting under observation, and {\em epistemic} uncertainty, which is the noise originating from the model. The aleatoric type of uncertainty corresponds to the uncertainty around $v_i$'s in our model. It depends on the nature of the domain, and not on the specific machine learning system used, and therefore might be harder to quantify, whereas the epistemic kind, corresponding to $a_i$'s in our model, is easier to estimate. Another motivating example is when values are observed through physical measurement~\cite{MeasurementError}, where $a_i$'s capture the error in the measurement instruments. Yet another interesting example is when values are observed through a differentially private mechanism~\cite{DPbook}, in which case noises $a_i$ are explicitly added to the values by the mechanism. In all these examples, estimating the distribution of noise is significantly easier than estimating the distribution of the values, and therefore it is natural to look for algorithms that work when noises are stochastic and the values are adversarially chosen. 

Our first observation is that the simple algorithm that always picks the item with the highest observed value can have a regret that is unboundedly worse than the optimal regret, even if the noises are unbiased. In fact, we will show that the optimal algorithm is sometimes non-deterministic and, even if we restrict to the class of deterministic algorithms, need not be monotone. On the plus side, we give an algorithm that always achieves a regret that is within a constant factor of the optimal regret. Our algorithm has a very simple form: instead of choosing the item with the highest observed value $v_i+a_i$, it chooses the one with the highest $v_i+a_i-\mathrm{offset}_i$, where $\mathrm{offset}_i$ is an {\em offset} value that only depends on the noise distribution $A_i$ and not on any other noise distribution. Furthermore, the offset values can be computed efficiently using a formula that is reminiscent of the monopoly pricing formula in economic theory. 

\paragraph{Organization.} The rest of this paper is organized as follows: We present related work in \cref{sec:related} and formally introduce our model and notations in \cref{sec:prelim}. Our results and some intuition about the proofs are stated in \cref{sec:results}. In \cref{sec:binary}, we prove the main result in a simpler \emph{binary setting} to introduce some of our proof ideas and derive a result for later use. We move on to the general model in \cref{sec:general} and present our full proof there. In \cref{sec:conclusions}, we give several examples that prove lower bounds on the performance of natural algorithms (\simple) or classes of algorithms (deterministic, monotone).
We conclude our discussion with some open questions in \cref{sec:conclusions}.

\section{Related Work}
\label{sec:related}
Our model can be considered as a mid-way point between robust and stochastic optimization. Robust optimization (see the survey \cite{BertsimasBC11} for details) studies settings where distributional information is not available, and aims to find solutions that work well for the worst-case set of values in an allowable space. Stochastic optimization, on the other hand, assumes the values come from known distributions. See, for example,~\cite{BirgLouv97,IKMM,Nikolova}. The way our model combines a worst-case instance with random noise is similar to the celebrated smoothed analysis model of Spielman and Teng~\cite{SpielmanT04}, which interpolates between worst-case and average-case analyses of the simplex algorithm by evaluating its expected running time on instances obtained by adding random noise to an adversarially chosen instance. 

Another component of our model is the notion of minimizing the worst-case regret, which was originally presented in \cite{Savage51}, and has been used in numerous contexts in theoretical computer science (see, for example,~\cite{bell1982regret,gordon1999regret,AuerMAB}).

The problem of finding the maximum (or in general top-$k$) with noisy information has been extensively studied in the noisy pairwise comparison model (see for example \cite{FeigeRPU94, DBLP:conf/icml/Busa-FeketeSCWH13, BMW16, CohenAddadMM20}). The goal there is usually minimizing the query complexity or the round complexity for outputting the maximum correctly with high probability. 

The area of prior-independent auctions (e.g. \cite{DBLP:conf/stoc/HartlineR08,DBLP:journals/geb/DhangwatnotaiRY15,DBLP:conf/sigecom/FuILS15,DBLP:journals/mansci/AllouahB20,DBLP:conf/focs/HartlineJL20}) is on designing near-optimal auctions for revenue maximization in a similar and harsher setting in which the prior information of bidders' values is completely unknown.

\section{Model}
\label{sec:prelim}
Let there be a set $N = \{1, 2, \ldots, n\}$ of $n$ items. Item $i$ has an intrinsic value $v_i \in \mathbb{R}$, and is associated with a noise $a_i \in \mathbb{R}$, that is drawn from a known distribution $A_i$. We assume all these draws are mutually independent. The observed value of this item is $s_i := v_i + a_i$. We assume, without loss of generality, that the noise distributions $A_i$'s are unbiased, i.e., $\E[A_i] = 0$ for every $i\in N$.\footnote{This is without loss of generality in our model, since the $A_i$ distributions are known and the algorithm can subtract $E[A_i]$ from the observed value $s_i$.}

An algorithm $\alg$ aims to select exactly one item from $N$ with a high intrinsic value. It has access to the observed values $\{s_i\}_{i \in N}$ and the noise distributions $\{A_i\}_{i \in N}$. We evaluate $\alg$ based on its regret. For any vector of intrinsic values, its regret is the expected difference (over the noise draws) in intrinsic value between the best item and the pick of $\alg$. Formally, let $\alg(\vec{s}, \vec{A}; \vec{v})$ be the expected intrinsic value that $\alg$ picks when its input is $(\vec{s}, \vec{A})$ and the intrinsic value vector is $\vec{v}$, and we have
\[
\reg(\alg, \vec{v}) := \max_{i \in N} v_i - \E_{\vec{a} \sim \vec{A}}\left[\alg(\vec{s} = \vec{v} + \vec{a}, \vec{A}; \vec{v})\right].
\]
The overall regret of $\alg$ is defined to be its regret on the worst-case choice of intrinsic values:
\[
\reg(\alg) := \sup_{\vec{v} \in \mathbb{R}^n} \reg(\alg, \vec{v}).
\]

Given the noise distributions $\vec{A}$, $\ben$ is the optimal regret that any algorithm achieves in the instance.
\[
\ben := \inf_{\alg} \reg(\alg).
\]

We say an algorithm $\alg$ is a $c$-approximation if for every set of noise distributions $\vec{A}$, $\reg(\alg)$ is at most $c\cdot\ben$. Note that for certain set of distributions $\vec{A}$, $\ben$ can be infinite. To avoid such cases, we make the benign assumption that for each $i \in N$, the function $\varphi_i(x) := \Pr_{a_i \sim A_i}[|a_i| \geq x] \cdot x$ is bounded on $x \in [0, +\infty)$. It is implied by, e.g., $\E[|A_i|]$ being finite.

Perhaps the simplest and the most natural algorithm for this problem is one that always picks the highest observed value. We call this algorithm $\simple$.

\paragraph{Binary case.}
An important special case of the problem is when $n=2$. We call this the {\em binary} case, since the algorithm only needs to make a binary choice between the two items. In the binary case, without loss of generality, we can assume that the second item has value $v_2 = 0$ and noise distribution $A_2$ that is deterministically $0$. This is because for any two distributions $A_1$ and $A_2$, solving the problem on the observed value $(s_1,s_2)$ and distributions $(A_1,A_2)$ is equivalent to solving the problem on $(s_1-s_2,0)$ and distributions $(A_1-A_2,0)$. Therefore, the binary case of the problem is equivalent to the problem of deciding whether to pick an item $1$ and get $v_1$, or not to do so and get $0$.

\section{Our Results}
\label{sec:results}
First, to illustrate that the problem is non-trivial and get an intuition for some of the hard cases, we observe that the simple algorithm $\simple$ can have a regret that is arbitrarily worse than $\ben$, even in the binary case.

\begin{restatable}{theorem}{expectation}
For any $M \in \mathbb{R}^+$, there is a distribution $A_1$ such that in the binary case with noise distribution $A_1$,
\[
\reg(\simple) > M \cdot \ben.
\]
\label{thm:expectation}
\end{restatable}

In fact, as we observe in the next two theorems, the optimal algorithm can have a complicated form even in the binary case: it is sometimes non-deterministic, and even if we restrict to deterministic algorithms, it can be non-monotone. The proofs of these theorems are presented in \cref{sec:lower}.

\begin{restatable}{theorem}{deterministic}
\label{thm:deterministic}
There is a distribution $A_1$ such that in the binary case with noise distribution $A_1$, for every deterministic $\alg$ and every $\ep>0$, $\reg(\alg) > (2-\ep)\cdot\ben$. 
\label{thm:sep_random}
\end{restatable}

\begin{restatable}{theorem}{monotone}
\label{thm:monotone}
For every $\ep > 0$, there is a distribution $A_1$ and a deterministic $\alg'$ for the binary case with noise distribution $A_1$, such that for every deterministic and monotone $\alg$, we have
\[
\reg(\alg) > (2 - \varepsilon) \cdot \reg(\alg'). 
\]
\end{restatable}

Note that in the binary case, a deterministic and monotone algorithm has a particularly simple form: it can be described by a simple threshold $t$ (which depends on the noise distribution $A_1$), and picks item 1 if and only if $s_1\ge t$. In fact, as we will observe in \cref{subsec:theta}, the optimal such threshold $t$ can be written in closed form. Unfortunately, the above results show that this simple algorithm can be sub-optimal. But could it be {\em approximately} optimal?

We prove that this is indeed the case. In fact, even for the non-binary case, we can prove that a generalization of this simple algorithm provides a constant approximation of regret. The algorithm, which we call $\off_\theta$, is parameterized by a function $\theta$ that converts each noise distribution $A_i$ into a real-valued {\em offset} $\theta_i := \theta(A_i)$. It then picks $\argmax_{i \in N} (s_i - \theta_i)$. 
Notice that $\simple$ is $\off_\theta$ with the choice of $\theta(A) =0$. It turns out that with a proper choice of function $\theta$, $\off_\theta$ gives a constant-approximation to $\ben$.

\begin{theorem}
\label{thm:main}
For the choice of $\theta$ described in \cref{subsec:theta}, we have
\[
\reg(\off_\theta) \leq O(1) \cdot \ben,
\]
for any $n$ and any choice of noise distributions $\vec{A}$.
\end{theorem}

The proof of the above theorem, which is our main result, is presented in \cref{sec:binary} in the binary case, and in \cref{sec:general} in the general case. The constant approximation factor we can prove in the binary case is $24$, and in the general case, whose proof uses the binary result, is $57000$. 

The algorithm $\off_\theta$ is almost as simple and efficient as $\simple$. It only needs to store one scalar, $\theta_i$, that summarizes all it needs to know about the noise distribution $A_i$, and does not need to consider any {\em interaction} between different noise distributions. With reasonable representations of the $A_i$ distributions (e.g., a list representation when the support is finite) it can be computed in linear time.


Also, \cref{thm:main} implies the following curious corollary: when all noise distributions are symmetric, even if they are wildly  heterogeneous, the plain $\simple$ algorithm achieves a constant approximation. 

\begin{corollary}
\label{cor:symm}
When all noise distributions are symmetric, $\simple$ is a constant-approximation to the optimal regret.
\end{corollary}

\paragraph{Technical Highlight.} Our main algorithm, $\off_\theta$, naturally arises in the discussion of the special case of our binary setting. However, it is a challenging task to show that it guarantees a constant approximation ratio. A recurring theme in our proof is to characterize the noise distributions using scalar quantities, and then design pairs of hard instances parameterized by those scalars where no algorithm can perfectly distinguish them because of the noises. The hard instances for the optimal algorithm use ideas of convolutions of functions to hide common parts of observed values and thus are indistinguishable by any algorithm. This in turn establishes that no algorithm can do much better than ours. We first look at a simpler $2$-item setting to gain intuition and obtain partial results. Then in the general case, we gradually analyze and simplify our problem until we reach a manageable formulation.

\section{The Binary Setting}
\label{sec:binary}
In this section, we prove our main result (\cref{thm:main}) in the binary case. We will later use this result for proving the general case. The proof in this section also help illustrate some ideas in the general proof.

\subsection{Choice of $\theta$}
\label{subsec:theta}
In the binary case, the value $\theta$ is a threshold, so that item 1 is picked if and only if the observed value $s_1$ is greater than this threshold. We will have a regret in two cases: when $v_1 >0$ but the algorithm does not pick 1, and when $v_1 < 0$ and the algorithm picks 1. The former case happens with probability $\Pr_{a_1\sim A_1}[v_1 + a_1\le \theta]$ and the magnitude of regret is $v_1$. The latter case happens with probability $\Pr_{a_1\sim A_1}[v_1 + a_1> \theta]$ and the magnitude of regret is $(-v_1)$. Therefore, the overall regret at threshold $\theta$ can be written as:

$$\max\left(\max_{v_1 > 0}\big(\Pr_{a_1\sim A_1}[a_1\le \theta-v_1] \cdot v_1\big), \max_{v_1 < 0}\big(\Pr_{a_1\sim A_1}[a_1> \theta-v_1] \cdot (-v_1)\big)\right).$$
We choose $\theta$ as the value that minimizes the above regret:
\[
\theta(D) := \argmin_{t \in \mathbb{R}} \max\left(\max_{v > 0}\big(\Pr[D \leq t - v] \cdot v\big), \max_{v < 0}\big(\Pr[D \geq t - v] \cdot (-v)\big)\right).
\]
We will use the notation $\theta(D)$ to refer to this specific choice in subsequent discussion.

Without loss of generality, we assume $\max_{v > 0}\big(\Pr[D \leq \theta(D) - v] \cdot v\big) = \max_{v < 0}\big(\Pr[D \geq \theta(D) - v] \cdot (-v)\big)$ to simplify future discussion. This naturally holds if $D$ is a continuous distribution, and we can slightly perturb $D$ if not.

We next show that our choice of $\theta$ gives a constant-approximation to $\ben$ in the binary setting.

\begin{theorem}
In the binary setting, we have
\[
\reg(\off_\theta) \leq 24 \cdot \ben.
\]
\label{thm:binary}
\end{theorem}

\subsection{Proof of \cref{thm:binary}}

Fix any noise distribution $D$. Let $v^+ := \argmax_{v > 0}\big(\Pr[D \leq t - v] \cdot v\big)$, and let $v^- := \argmax_{v < 0}\big(\Pr[D \geq t - v] \cdot (-v)\big)$. \cref{lem:tail} states tail bounds for $D$: It must be somewhat concentrated around $\theta(D) - v^+$ and $\theta(D) - v^-$.

\begin{lemma}
For any $\lambda > 1$, $\Pr_{a \sim D}[a \leq \theta(D) - \lambda \cdot v^+ \mid a \leq \theta(D) - v^+] \leq \frac{1}{\lambda}$. Similarly, for any $\lambda > 1$, $\Pr_{a \sim D}[a \geq \theta(D) - \lambda \cdot v^- \mid a \geq \theta(D) - v^-] \leq \frac{1}{\lambda}$. 
\label{lem:tail}
\end{lemma}
\begin{proof}
By definition of $v^+$, we know
\[
\Pr_{a \sim D}[a \leq \theta(D) - v^+] \cdot v^+ \geq \Pr_{a \sim D}[a \leq \theta(D) - \lambda \cdot v^+] \cdot \lambda \cdot v^+.
\]
Rearranging gives the first statement.

Similarly, by definition of $v^-$, we know
\[
\Pr_{a \sim D}[a \geq \theta(D) - v^-] \cdot (-v^-) \geq \Pr_{a \sim D}[a \geq \theta(D) - \lambda \cdot v^-] \cdot (-\lambda \cdot v^-).
\]
Rearranging gives the second statement.
\end{proof}

Denote $U[\ell, r]$ to be the uniform distribution supported on the interval $[\ell, r]$. In \cref{lem:uniform}, we show that the noise plus a carefully-chosen uniformly-distributed value can ``hide'' another uniform distribution, which we will ultimately utilize to show no algorithm can distinguish between difference cases and thus give a lower bound of $\ben$ in \cref{lem:distinguish}.

\begin{lemma}
Fix any noise distribution $A$. Let $V$ be the uniform distribution $U[0.5v^+, 1.5v^+]$. Then the probability density function (PDF) of $V + A$ satisfies:
\[
f_{V + A}(t) \geq \frac{1}{3} \cdot \frac{1}{v^+} \cdot \Pr_{a \sim A}[a < \theta(A) - v^+]
\]
for any $t \in [\theta(A) - 0.5v^+, \theta(A)]$.

Similarly, let $V'$ be the uniform distribution $U[1.5v^-, 0.5v^-]$. Then the PDF of $V' + A$ satisfies:
\[
f_{V' + A}(t) \geq \frac{1}{3} \cdot \frac{1}{-v^-} \cdot \Pr_{a \sim A}[a > \theta(A) - v^-]
\]
for any $t \in [\theta(A), \theta(A) - 0.5v^-]$.
\label{lem:uniform}
\end{lemma}
\begin{proof}
We prove the first statement and a similar argument works for the second.

For any $t \in [\theta(A) - 0.5v^+, \theta(A)]$,
\begin{align*}
f_{V + A}(t) &= \int_{0.5v^+}^{1.5v^+} f_V(q) \cdot f_A(t - q) \d q\\
&= \frac{1}{v^+} \cdot \int_{0.5v^+}^{1.5v^+} f_A(t - q) \d q\\
&= \frac{1}{v^+} \cdot \int_{t - 1.5v^+}^{t - 0.5v^+} f_A(q) \d q\\
&\geq \frac{1}{v^+} \cdot \int_{\theta(A) - 1.5v^+}^{\theta(A) - v^+} f_A(q) \d q\\
&\geq \frac{1}{v^+} \cdot \left(1 - \frac{1}{1.5}\right) \cdot \Pr[a < \theta(A) - v^+],
\end{align*}
where we used \cref{lem:tail} in the last step.
\end{proof}

Without loss of generality, assume $v^+ \leq -v^-$ for the noise distribution $A_1$. (Otherwise we can flip $A_1$.)
\begin{lemma}
Define $k := \left\lceil \frac{-v^-}{v^+} \right\rceil$. We have $\ben \geq \frac{1}{24} \cdot \reg(\off_\theta)$ on either of the following two instances:
\begin{enumerate}
    \item $v_1$ is drawn from $V$, which is the uniform average of the following distributions: $U[v^+, 2v^+]$, $U[1.5v^+, 2.5v^+]$, $\ldots$, $U[(0.5 + 0.5k)v^+$, $(1.5 + 0.5k)v^+]$.
    \item $v_1$ is drawn from $V' = U[1.5v^-, 0.5v^-]$.
\end{enumerate}
\label{lem:distinguish}
\end{lemma}
\begin{proof}
Fix any algorithm $\alg$. According to \cref{lem:uniform}, if $v_1$ is drawn from $U[(0.5 + 0.5t)v^+, (1.5 + 0.5t)v^+]$, the distribution of $s_1 = a_1 + v_1$ will contain a uniform distribution on $[\theta(A_1) + 0.5(t - 1)v^+, \theta(A_1) + 0.5tv^+]$. Summing $t$ from $1$ to $k$, the intervals together cover $[\theta(A_1), \theta(A_1) - 0.5v^-]$.

In light of this, denote $U[\theta(A_1), \theta(A_1) - 0.5v^-]$ as $S$, and evaluate $\reg(\alg)$ with $v_1 \sim V$:
\[
\reg(\alg) \geq \Pr_{s_1 \sim S}[\alg \text{ picks } v_2 \text{ on input } s_1] \cdot 1.0v^+ \cdot (-0.5v^-) \cdot \frac{1}{k} \cdot \frac{1}{3} \cdot \frac{1}{v^+} \cdot \Pr[a_1 < \theta(A_1) - v^+],
\]
where $1.0v^+$ is the minimum loss whenever $\alg$ makes an incorrect decision, and $(-0.5v^-) \cdot \frac{1}{k} \cdot \frac{1}{3} \cdot \frac{1}{v^+} \cdot \Pr[a_1 < \theta(A_1) - v^+]$ is the adjusted probability given by \cref{lem:uniform}.

Similarly, evaluate $\reg(\alg)$ with $v_1 \sim V'$ and we get
\[
\reg(\alg) \geq \Pr_{s_1 \sim S}[\alg \text{ picks } v_1 \text{ on input } s_1] \cdot (-0.5v^-) \cdot (-0.5v^-) \cdot \frac{1}{k} \cdot \frac{1}{3} \cdot \frac{1}{-v^-} \cdot \Pr[a_1 > \theta(A_1) - v^-],
\]

Let $p_1 := \Pr_{s_1 \sim S}[\alg \text{ picks } v_1 \text{ on input } s_1]$, and we have
\begin{align*}
\reg(\alg) &\geq \max \left(\frac{1}{6} \cdot (-v^-) \cdot \frac{1}{k} \cdot \Pr[a_1 < \theta(A_1) - v^+] \cdot (1 - p_1), \ \frac{1}{12} \cdot (-v^-) \cdot \Pr[a_1 > \theta(A_1) - v^-] \cdot p_1\right)\\
&\geq \max \left(\frac{1}{12} \cdot v^+ \cdot \Pr[a_1 < \theta(A_1) - v^+] \cdot (1 - p_1), \ \frac{1}{12} \cdot (-v^-) \cdot \Pr[a_1 > \theta(A_1) - v^-] \cdot p_1\right)\\
&\geq \frac{1}{12} \cdot \max \left(\reg(\off_\theta) \cdot (1 - p_1), \ \reg(\off_\theta) \cdot p_1\right)\\
&\geq \frac{1}{24} \cdot \reg(\off_\theta). \qedhere
\end{align*}
\end{proof}

\section{The General Case}
\label{sec:general}
In this section, we move on to the general case and show that $\off_\theta$ remains constant-competitive against any algorithm.
\subsection{Reduction: Highest Value Has No Noise}
We first argue that if we solve the case where the item with the highest value is associated with no noise, then we can solve the general case.

Consider the vector of values $(\hat v_1, \ldots, \hat v_n)$ that maximizes the regret of $\off_\theta$. Assume $\hat v_1$ is the largest among them via renaming the items. Recall that we denote $\reg(\alg, (A_1, \ldots, A_n), (v_1, \ldots, v_n))$ to be the regret of $\alg$ on the value vector $(v_1, \ldots, v_n)$. We have
\begin{align*}
&\reg\big(\off_\theta, (A_1, \ldots, A_n)\big)\\
= &\reg\big(\off_\theta, (A_1, \ldots, A_n), (\hat v_1, \ldots, \hat v_n)\big)\\
= &\sum_{i = 2}^n \Pr_{\Vec{a} \sim \Vec{A}}\big[\hat v_i + a_i - \theta_i \geq \hat v_1 + a_1 - \theta_1 \land \hat v_i + a_i - \theta_i = \max_{j = 2}^n \hat v_j + a_j - \theta_j\big] \cdot (\hat v_1 - \hat v_i).
\end{align*}
By the union bound, we further get
\begin{align*}
&\reg\big(\off_\theta, (A_1, \ldots, A_n)\big)\\
\leq &\sum_{i = 2}^n \bigg(\Pr_{\Vec{a} \sim \Vec{A}}\big[\hat v_i + a_i - \theta_i \geq \frac{\hat v_1 + \hat v_i}{2} \land \hat v_i + a_i - \theta_i = \max_{j = 2}^n \hat v_j + a_j - \theta_j\big] \\
&\quad + \Pr_{\Vec{a} \sim \Vec{A}}\big[\frac{\hat v_1 + \hat v_i}{2} \geq \hat v_1 + a_1 - \theta_1 \land \hat v_i + a_i - \theta_i = \max_{j = 2}^n \hat v_j + a_j - \theta_j\big]\bigg) \cdot (\hat v_1 - \hat v_i)\\
= &\sum_{i = 2}^n \bigg(\Pr_{\Vec{a} \sim \Vec{A}}\big[\frac{\hat v_i}{2} + a_i - \theta_i \geq \frac{\hat v_1}{2} \land \hat v_i + a_i - \theta_i = \max_{j = 2}^n \hat v_j + a_j - \theta_j\big] \\
&\quad + \Pr_{\Vec{a} \sim \Vec{A}}\big[\frac{\hat v_i}{2} \geq \frac{\hat v_1}{2} + a_1 - \theta_1 \land \hat v_i + a_i - \theta_i = \max_{j = 2}^n \hat v_j + a_j - \theta_j\big]\bigg) \cdot (\hat v_1 - \hat v_i).
\end{align*}
We bound this using the following two lemmas.
\begin{lemma}
\label{lem:highest_no_noise_1}
$\sum_{i = 2}^n \Pr_{\Vec{a} \sim \Vec{A}}\big[\frac{\hat v_i}{2} + a_i - \theta_i \geq \frac{\hat v_1}{2} \land \hat v_i + a_i - \theta_i = \max_{j = 2}^n \hat v_j + a_j - \theta_j\big] \cdot (\hat v_1 - \hat v_i) \leq \sum_{i = 2}^n \Pr_{\Vec{a} \sim \Vec{A}}\big[\frac{\hat v_i}{2} + a_i - \theta_i \geq \frac{\hat v_1}{2} \land \frac{\hat v_i}{2} + a_i - \theta_i = \max_{j = 2}^n \frac{\hat v_j}{2} + a_j - \theta_j\big] \cdot (\hat v_1 - \hat v_i).$
\end{lemma}
\begin{proof}
Denote $p = \argmax_{j = 2}^n \hat v_j + a_j - \theta_j$, and $q = \argmax_{j = 2}^n \frac{\hat v_j}{2} + a_j - \theta_j$. We have
\[
\mathrm{LHS} = \E_{\Vec{a} \sim \Vec{A}}\left[\mathds{1}\left(\frac{\hat v_p}{2} + a_p - \theta_p \geq \frac{\hat v_1}{2}\right) \cdot (\hat v_1 - \hat v_p)\right],
\]
and
\[
\mathrm{RHS} = \E_{\Vec{a} \sim \Vec{A}}\left[\mathds{1}\left(\frac{\hat v_q}{2} + a_q - \theta_q \geq \frac{\hat v_1}{2}\right) \cdot (\hat v_1 - \hat v_q)\right].
\]
Fix the realized noises $\Vec{a}$. If $p = q$, then $\mathrm{LHS} = \mathrm{RHS}$ for that realization. Otherwise, notice that we must have $\hat v_p \geq \hat v_q$, and thus $\hat v_1 - \hat v_p \leq \hat v_1 - \hat v_q$. Additionally, by the definition of $q$, we have
\[
\mathds{1}\left(\frac{\hat v_q}{2} + a_q - \theta_q \geq \frac{\hat v_1}{2}\right) \geq \mathds{1}\left(\frac{\hat v_p}{2} + a_p - \theta_p \geq \frac{\hat v_1}{2}\right).
\]
Therefore, $\mathrm{LHS} \leq \mathrm{RHS}$ in the lemma statement.
\end{proof}
\begin{lemma}
\label{lem:highest_no_noise_2}
$\sum_{i = 2}^n \Pr_{\Vec{a} \sim \Vec{A}}\big[\frac{\hat v_i}{2} \geq \frac{\hat v_1}{2} + a_1 - \theta_1 \land \hat v_i + a_i - \theta_i = \max_{j = 2}^n \hat v_j + a_j - \theta_j\big] \cdot (\hat v_1 - \hat v_i) \leq \max_{i = 2}^n \Pr_{\Vec{a} \sim \Vec{A}}\big[\frac{\hat v_i}{2} \geq \frac{\hat v_1}{2} + a_1 - \theta_1\big] \cdot (\hat v_1 - \hat v_i).$
\end{lemma}
\begin{proof}
Notice that the choice of $\argmax_{j = 2}^n \hat v_j + a_j - \theta_j$ is independent from the realization of $a_1 \sim A_1$. We therefore have
\begin{align*}
\mathrm{LHS} &= \sum_{i = 2}^n \Pr_{\Vec{a} \sim \Vec{A}}\big[\frac{\hat v_i}{2} \geq \frac{\hat v_1}{2} + a_1 - \theta_1\big] \cdot \Pr\big[\hat v_i + a_i - \theta_i = \max_{j = 2}^n \hat v_j + a_j - \theta_j\big] \cdot (\hat v_1 - \hat v_i)\\
&\leq \mathrm{RHS}. \qedhere
\end{align*}
\end{proof}

By \cref{lem:highest_no_noise_1} and \cref{lem:highest_no_noise_2}, we get
\begin{align*}
&\reg\big(\off_\theta, (A_1, \ldots, A_n)\big)\\
\leq &\sum_{i = 2}^n \bigg(\Pr_{\Vec{a} \sim \Vec{A}}\big[\frac{\hat v_i}{2} + a_i - \theta_i \geq \frac{\hat v_1}{2} \land \frac{\hat v_i}{2} + a_i - \theta_i = \max_{j = 2}^n \frac{\hat v_j}{2} + a_j - \theta_j\big] \cdot (\hat v_1 - \hat v_i)\bigg)\\
&\quad + \max_{i = 2}^n \Pr_{\Vec{a} \sim \Vec{A}}\big[\frac{\hat v_i}{2} \geq \frac{\hat v_1}{2} + a_1 - \theta_1\big] \cdot (\hat v_1 - \hat v_i)\\
\leq &2 \cdot \reg(\off_\theta, (0, A_2, \ldots, A_n), (\hat v_1 / 2, \ldots, \hat v_n / 2)) + 2 \cdot \reg(\off_\theta, (A_1, 0)).
\end{align*}
Notice that $\hat v_1 / 2$ is still the largest $\hat v_i / 2$. Now we only need to prove
\begin{itemize}
    \item If $v_1$ is the largest in $\alg$, then
\[
\reg(\off_\theta, (0, A_2, \ldots, A_n), (v_1, v_2, \ldots, v_n)) = O(1) \cdot \reg(\opt, (A_1, \ldots, A_n)).
\]
    \item $\reg(\off_\theta, (A_1, 0)) = O(1) \cdot \reg(\opt, (A_1, \ldots, A_n))$.
\end{itemize}
The second point is immediate, as $\reg(\off_\theta, (A_1, 0)) \leq 24 \cdot \reg(\opt, (A_1, 0))$ by \cref{thm:binary}, which is in turn at most $24 \cdot \reg(\opt, (A_1, A_2))$ since an algorithm for $(A_1, 0)$ can add noises and simulate $(A_1, A_2)$. This is in turn at most $24 \cdot \reg(\opt, (A_1, \ldots, A_n))$, since the values for Items $3, \ldots, n$ can be arbitrarily negative.

For the first point, notice that
\[
\reg(\opt, (0, A_2, \ldots, A_n)) \leq \reg(\opt, (A_1, A_2, \ldots, A_n)),
\]
since $\opt$ can simulate a noise from $A_1$ for the first item. Therefore, we only need that if $v_1$ is the largest among $v_i$'s, then
\[
\reg(\off_\theta, (0, A_2, \ldots, A_n), (v_1, v_2, \ldots, v_n)) = O(1) \cdot \reg(\opt, (0, A_2, \ldots, A_n)).
\]
We will focus on proving this in the rest of this section.

\subsection{Linearization}
Without loss of generality, fix $v_1 = 0$ and impose the constraints that $v_i \leq 0$ for every $i \neq 1$. Let $(v_2^*, \ldots, v_n^*)$ be the solution to maximize
\[
\sum_{i = 2}^n \Pr_{a_i \sim A_i}[v_i + a_i - \theta_i \geq 0] \cdot (-v_i), \text{\quad  s.t. } \sum_{i = 2}^n \Pr_{a_i \sim A_i}[v_i + a_i - \theta_i \geq 0] \leq 0.5.
\]

Slightly abusing the notation, let $(\hat v_2, \ldots \hat v_n)$ be the value vector that maximizes $\reg(\off_\theta, (0, A_2, \ldots, A_n))$. We can uniformly and gradually reduce $(\hat v_2, \ldots \hat v_n)$ until reaching $(v'_2, \ldots, v'_n)$ that satisfies $\Pr_{\Vec{a} \sim \Vec{A}}[\exists i \in \{2, \ldots, n\}, \ v_i' + a_i - \theta_i \geq 0] \leq \frac{1}{2.55} < 1 - \frac{1}{\sqrt{e}}$. This implies $\sum_{i = 2}^n \Pr_{a_i \sim A_i}[v_i + a_i - \theta_i \geq 0] \leq 0.5$ and thus $(v'_2, \ldots, v'_n)$ is feasible to the aforementioned program. We therefore have
\begin{align*}
    \reg(\off_\theta, (0, A_2, \ldots, A_n)) &\leq 2.55 \cdot \reg(\off_\theta, (0, A_2, \ldots, A_n), (v_2', \ldots, v_n'))\\
    &\leq 2.55 \cdot \sum_{i = 2}^n \Pr_{a_i \sim A_i}[v_i' + a_i - \theta_i \geq 0] \cdot (-v_i')\\
    &\leq 2.55 \cdot \sum_{i = 2}^n \Pr_{a_i \sim A_i}[v_i^* + a_i - \theta_i \geq 0] \cdot (-v_i^*).
\end{align*}

We therefore only need to prove
\[
\reg(\opt, (0, A_2, \ldots, A_n)) = \Omega(1) \cdot \sum_{i = 2}^n \Pr_{a_i \sim A_i}[v_i^* + a_i - \theta_i \geq 0] \cdot (-v_i^*).
\]

\subsection{Completing the Proof}
Define $b = \sum_{i = 2}^n \Pr_{a_i \sim A_i}[v_i^* + a_i - \theta_i \geq 0] \cdot (-v_i^*)$. Define $I = \{i \in \{2, 3, \ldots, n\} \mid -v_i^* \geq b\}$. We have the following lemma.
\begin{lemma}
\label{lem:multi_Ib}
$\frac{b}{2} \leq \sum_{i \in I} \Pr_{a_i \sim A_i}[v_i^* + a_i - \theta_i \geq 0] \cdot (-v_i^*) \leq b$.
\end{lemma}
\begin{proof}
The second inequality is immediate.
The first inequality is because
\begin{align*}
\sum_{i \in I} \Pr_{a_i \sim A_i}[v_i^* + a_i - \theta_i \geq 0] \cdot (-v^*_i) &= b - \sum_{i \in \{2, \ldots, n\} \setminus I} \Pr_{a_i \sim A_i}[v_i^* + a_i - \theta_i \geq 0] \cdot (-v^*_i)\\
&\geq b - \sum_{i \in \{2, \ldots, n\} \setminus I} \Pr_{a_i \sim A_i}[v_i^* + a_i - \theta_i \geq 0] \cdot b\\
&\geq b - \sum_{i \in \{2, \ldots, n\}} \Pr_{a_i \sim A_i}[v_i^* + a_i - \theta_i \geq 0] \cdot b\\
&\geq 0.5b. \qedhere
\end{align*}
\end{proof}

Notice that if some $i \in I$ satisfies $\Pr_{a_i \sim A_i}[a_i \in [\theta_i - 0.5 v_i^*, \theta_i - v_i^*]] \geq 0.2$, then
\begin{align*}
\reg(\opt, (0, A_2, \ldots, A_n)) &\geq \reg(\opt, (0, A_i))\\
&\geq \frac{1}{24} \cdot \reg(\off_\theta, (0, A_i))\tag{By \cref{thm:binary}}\\
&\geq \frac{1}{24} \cdot \reg(\off_\theta, (0, A_i), (0, 0.5v_i^*))\\
&= \frac{1}{24} \cdot \Pr_{a_i \sim A_i}[0.5v_i^* + a_i > \theta_i] \cdot (-0.5 v_i^*)\\
&\geq \frac{1}{240} \cdot (-v_i^*) \geq \frac{1}{240} \cdot b,
\end{align*}
and we are done. Therefore, from now on, we assume that no $i \in I$ satisfies $\Pr_{a_i \sim A_i}[a_i \in [\theta_i - 0.5 v_i^*, \theta_i - v_i^*]] \geq 0.2$.

\begin{lemma}
\label{lem:multi_tail1}
For any $\lambda > 1$, $\Pr_{a_i \sim A_i}[a_i \geq \theta_i - \lambda \cdot v_i^* \mid a_i \geq \theta_i - v_i^*] \leq \frac{1}{\lambda}$.
\end{lemma}
\begin{proof}
It comes from the optimality of $v_i^*$. Replacing it with $\lambda v_i^*$ is still feasible, and thus should not improve the objective.
\end{proof}

\cref{lem:multi_tail1} immediately gives:
\begin{lemma}
\label{lem:multi_tail2}
$\Pr_{a_i \sim A_i}[a_i \in [\theta_i - v_i^*, \theta_i -1.1\cdot v_i^*]] \geq 0.09 \cdot \Pr[a_i \geq \ \theta_i - v_i^*]$.
\end{lemma}

\begin{lemma}
\label{lem:multi_center1}
For $i \in I$, $\Pr_{a_i \sim A_i}[a_i \in [\theta_i + 5 v_i^*, \theta_i - 0.5v_i^*]] \geq 0.1$.
\end{lemma}
\begin{proof}
Firstly, $\Pr_{a_i \sim A_i}[a_i > \theta_i - v_i^*] \leq 0.5$ by the choice of $v_i^*$. (We have $\sum_{i = 2}^n \Pr_{a_i \sim A_i}[a_i > \theta_i - v_i^*] \leq 0.5$.)

Next, notice that $\Pr_{a_i \sim A_i}[a_i < \theta_i + 5v_i^*] \leq 0.2$ by the choice of $\theta_i$. Otherwise, $\Pr_{a_i \sim A_i}[-5v_i^* + a_i < \theta_i] > 0.2$, giving a regret of more than $-v_i^*$ if it is against a zero-noise value. However, for negative values, the regret is less than $-v_i^*$ by the optimality of $v_i^*$. However, $\theta_i$ should have balanced the regrets instead.

Finally, $\Pr_{a_i \sim A_i}[a_i \in [\theta_i - 0.5v_i^*, \theta_i - v_i^*]] < 0.2$ by our earlier assumption for $i \in I$.
\end{proof}

\begin{lemma}
\label{lem:multi_center2}
For any $i \in I$, there exists some $k_i \in [-0.4, 5]$, so that
\[
\Pr_{a_i \sim A_i}[a_i \in [\theta_i + k_i \cdot v_i^*, \theta_i + (k_i-0.1) \cdot v_i^*]] \geq \frac{1}{550}.
\]
\end{lemma}
\begin{proof}
This is from the pigeonhole principle and \cref{lem:multi_center1}.
\end{proof}

\begin{lemma}
\label{lem:multi_distinguish1}
Let $V_i$ be the uniform distribution $U[1.2 \cdot v_i^*, v_i^*]$. The PDF of $V_i + A_i$ satisfies:
\[
f_{V_i + A_i}(t) \geq 0.045 \cdot \frac{1}{0.1v_i^*} \cdot \Pr_{a_i \sim A_i}[a_i \geq \theta_i - v_i^*]
\]
for any $t \in [\theta_i + 0.1v_i^*, \theta_i]$.
\end{lemma}
\begin{proof}
Fix any $t \in [\theta_i + 0.1v_i^*, \theta_i]$. We have
\begin{align*}
f_{V_i + A_i}(t) &= \int_{1.2 \cdot v_i^*}^{v_i^*} \frac{1}{0.2 v_i^*} \cdot f_{A_i}(t - q) \d q\\
&\geq \int_{t - \theta_i + 1.1 \cdot v_i^*}^{t - \theta_i + v_i^*} \frac{1}{0.2 v_i^*} \cdot f_{A_i}(t - q) \d q\\
&= \frac{1}{0.2 v_i^*} \cdot \int_{\theta_i - v_i^*}^{\theta_i - 1.1 \cdot v_i^*} f_{A_i}(q) \d q\\
&\geq \frac{1}{0.2 v_i^*} \cdot 0.09 \cdot \Pr_{a_i \sim A_i}[a_i \geq \theta_i - v_i^*]\tag{By \cref{lem:multi_tail2}}\\
&= \frac{1}{0.1 v_i^*} \cdot 0.045 \cdot \Pr_{a_i \sim A_i}[a_i \geq \theta_i - v_i^*].
\end{align*}
This is exactly what we need.
\end{proof}

Similar to \cref{lem:multi_distinguish1}, we have the following lemma:
\begin{lemma}
\label{lem:multi_distinguish2}
Let $V'_i$ be the uniform distribution $U[(-k_i + 0.2) \cdot v_i^*, (-k_i) \cdot v_i^*]$, where $k_i$ is from \cref{lem:multi_center2}. The PDF of $V'_i + A_i$ satisfies:
\[
f_{V'_i + A_i}(t) \geq 0.0009 \cdot \frac{1}{0.1v_i^*}
\]
for any $t \in [\theta_i + 0.1v_i^*, \theta_i]$.
\end{lemma}

Now consider the following hard instance for $\opt$:
\begin{itemize}
    \item $v_1$ is known to be $0$.
    \item $v_i$ is ``typically'' drawn from $U[0.4 \cdot v_i^*, 0.2 \cdot v_i^*]$. ($v_i \leq 0.2 \cdot v_i^*$ in this case.)
    \item With probability $\Pr_{a_i \sim A_i}[v_i^* + a_i - \theta_i \geq 0]$ in a disjoint way, $v_i$ is instead ``atypically'' drawn from $U[(-k_i - 0.6) \cdot v_i^*, (-k_i - 0.8) \cdot v_i^*]$. ($v_i \geq -0.2 \cdot v_i^*$ in this case.)
\end{itemize}

\begin{lemma}
\label{lem:multi_distinguish3}
$\reg(\opt, (0, A_2, \ldots, A_n)) \geq 1.8 \times 10^{-4} \cdot \sum_{i \in I} \Pr_{a_i \sim A_i}[v_i^* + a_i - \theta_i \geq 0] \cdot (-v_i^*)$ for the instance above.
\end{lemma}
\begin{proof}
Let $p_i$ be the probability that $\opt$ chooses $v_i$ conditioned on the instance being ``typical''. Let $q_i$ be the probability that $\opt$ chooses $v_i$ conditioned on $v_i$ being ``atypical''. We have
\[
\reg(\opt, (0, A_2, \ldots, A_n)) \geq \sum_{i \in I} \left(\frac{1}{2} \cdot p_i \cdot (-0.2v_i^*) + \Pr_{a_i \sim A_i}[v_i^* + a_i - \theta_i \geq 0] \cdot (1 - q_i) \cdot (-0.2v_i^*)\right).
\]
Let $r_i$ be the probability of $\opt$ picking Item $i$ if $v_i + a_i \sim U[\theta_i - 0.7v_i^*, \theta_i - 0.8v_i^*]$. By \cref{lem:multi_distinguish1} and \cref{lem:multi_distinguish2}, we have
\[
p_i \geq r_i \cdot 0.045 \cdot \Pr[a_i \geq \theta_i - v_i^*], \text{ and } 1 - q_i \geq (1 - r_i) \cdot 0.0009.
\]
Therefore,
\begin{align*}
&\reg(\opt, (0, A_2, \ldots, A_n))\\
\geq &\sum_{i \in I} \left(\frac{1}{2} \cdot r_i \cdot 0.045 \cdot (-0.2v_i^*) + (1 - r_i) \cdot 0.0009 \cdot (-0.2v_i^*)\right) \cdot \Pr_{a_i \sim A_i}[a_i \geq \theta_i - v_i^*]\\
\geq &1.8 \times 10^{-4} \cdot \sum_{i \in I} \Pr[a_i \geq \theta_i - v_i^*] \cdot (-v_i^*). \qedhere
\end{align*}
\end{proof}

To summarize, we have the following theorem:
\begin{theorem}
$\off_\theta$ gives a $57000$-approximation in regret for the general setting.
\end{theorem}
\begin{proof}
\cref{lem:multi_distinguish3} shows that if no $i \in I$ satisfies $\Pr_{a_i \sim A_i}[a_i \in [\theta_i - 0.5 v_i^*, \theta_i - v_i^*]] \geq 0.2$, then
\[
\reg(\opt, (0, A_2, \ldots, A_n)) \geq 1.8 \times 10^{-4} \cdot \sum_{i \in I} \Pr_{a_i \sim A_i}[v_i^* + a_i - \theta_i \geq 0] \cdot (-v_i^*).
\]
Further, by \cref{lem:multi_Ib}, we have
\[
\reg(\opt, (0, A_2, \ldots, A_n)) \geq 9 \times 10^{-5} \cdot b.
\]
We have argued that if some $i \in I$ satisfies $\Pr_{a_i \sim A_i}[a_i \in [\theta_i - 0.5 v_i^*, \theta_i - v_i^*]] \geq 0.2$, then
\[
\reg(\opt, (0, A_2, \ldots, A_n)) \geq \frac{1}{240} \cdot b.
\]
Therefore, without the assumption on whether this is true, we always have
\[
\reg(\opt, (0, A_2, \ldots, A_n)) \geq 9 \times 10^{-5} \cdot b.
\]
On the other hand, we have shown
\[
\reg(\off_\theta, (0, A_2, \ldots, A_n)) \leq 2.55 \cdot b.
\]
To summarize the discussion in this section, we have
\begin{align*}
&\reg(\off_\theta, (A_1, A_2, \ldots, A_n))\\
\leq &2 \cdot \reg(\off_\theta, (0, A_2, \ldots, A_n)) + 2 \cdot \reg(\off_\theta, (A_1, 0))\\
\leq &5.1 \cdot b + 48 \cdot \reg(\opt, (A_1, \ldots, A_n))\\
\leq &\left(\frac{5.1}{9 \times 10^{-5}} + 48\right) \cdot \reg(\opt, (A_1, \ldots, A_n))\\
\leq &57000 \cdot \reg(\opt, (A_1, \ldots, A_n)). \qedhere
\end{align*}
\end{proof}

For any symmetric distribution $D$ (i.e., when $\Pr[D > x] = \Pr[D<-x]$ for every $x$), it is not hard to see that by the definition in \cref{subsec:theta}, $\theta(D)=0$. Therefore, in the event where each noise distribution is symmetric, 
$\off_\theta$ is the same as $\simple$, which simply outputs the highest observed value.
\begin{corollary}
$\simple$ gives a $57000$-approximation in regret when the noise distributions are symmetric.
\end{corollary}

\section{Lower Bound Examples}
\label{sec:lower}

In this section, we give a number of examples that show that a natural algorithm (\simple) and two natural classes of algorithms (deterministic, monotone) can be sub-optimal in our regret minimization model. In the case of \simple, the sub-optimality is by an arbitrary factor, and in the other two cases, by a factor of $2-\ep$. This proves \cref{thm:expectation,thm:deterministic,thm:monotone}.

We start with \cref{thm:expectation}. The example here is inspired by the super-constant gap between welfare and revenue on an equal-revenue distribution $D_{\mathrm{ER}}$ in pricing problems, where the optimal welfare $\E[D_{\mathrm{ER}}]$ is much larger than the optimal revenue $\max_p p \cdot \Pr[D_{\mathrm{ER}} \geq p]$.

\expectation*
\begin{proof}
Let $c > 1$ be a parameter that we fix later. Define $A_1$ so that
\[
\Pr[A_1 \leq t] =
\begin{cases}
    0       & \quad \text{if } t \in (-\infty, -1)\\
    \frac{1 + \ln c}{2 + \ln c}  & \quad \text{if } t \in [-1, 1)\\
    1 - \frac{1}{(2 + \ln c) \cdot t}  & \quad \text{if } t \in [1, c)\\
    1  & \quad \text{if } t \in [c, +\infty)
\end{cases}.
\]
One can check that
\begin{align*}
\E[A_1] &= \Pr[A_1 = -1] \cdot (-1) + \int_0^c \Pr[A_1 \geq t] \d t\\
&= -\frac{1 + \ln c}{2 + \ln c} + \frac{1}{2 + \ln c} + \int_1^c \frac{1}{(2 + \ln c) \cdot t} \d t\\
&= 0.
\end{align*}

Clearly, $\theta(A_1) \in [-1, 0]$ and it satisfies:
\[
\Pr[A_1 = -1] \cdot (\theta(A_1) - (-1)) = \max_{v < 0} \Pr[A_1 \geq \theta(A_1) - v] \cdot (-v) = \Pr[A_1 \geq 1] \cdot (1 - \theta(A_1)).
\]
Therefore,
\[
\frac{1 + \ln c}{2 + \ln c} \cdot (\theta(A_1) + 1) = \frac{1}{2 + \ln c} \cdot (1 - \theta(A_1)),
\]
and thus,
\[
\theta(A_1) = -\frac{\ln c}{2 + \ln c}.
\]
Further,
\[
\reg(\off_\theta) = \Pr[A_1 = -1] \cdot (\theta(A_1) - (-1)) = \frac{1 + \ln c}{2 + \ln c} \cdot \frac{2}{2 + \ln c}.
\]

On the other hand,
\[
\reg(\simple) = \Pr[A_1 = -1] \cdot 1 = \frac{1 + \ln c}{2 + \ln c}.
\]

Setting $c > e^{2M}$ concludes the proof.
\end{proof}

The next two results show the limitations of deterministic and monotone algorithms. Note that since $\off_\theta$ is both deterministic and monotone, these results also show that there exists an instance where $\reg(\off_\theta) \geq 2 \cdot \ben$.

\deterministic*

\begin{proof}
Define $A_1$ to be uniformly distributed on two values $\{-1, 1\}$. For any deterministic algorithm $\alg$, it picks either Item 1 or 2 when observe $s_1 = 0$. Without loss of generality we assume $\alg$ picks Item 1 when observe $s_1 = 0$. We set $v_1 = -1$. $\alg$ observes $s_1 = 0$ and pick Item 1 with probability $1/2$. So $\reg(\alg) \geq 1/2$.

Now consider a randomized $\alg'$ which picks Item 1 with probability $\min((s_1 + 1)/2, 1)$.  Now we bound $\reg(\alg', v)$ for different cases of $v$'s:
\begin{itemize}
\item For any $v_1 \geq 2$, we know this randomized algorithm always picks Item 1 and incurs no regret, i.e. $\reg(\alg', v) = 0$.
\item For any $v_1 \in [0, 2)$, we know the algorithm picks Item 2 only when $A_2 = -1$, and the probability of picking Item 2 is $1 - (v_1 - 1 + 1)/2$. Therefore, $\reg(\alg', v) = \frac{1}{2} \cdot v_1 \cdot (1 - v_1 / 2) \leq 1/4$.
\item For any $v_1 < 0$, we can use a similar argument as the previous two cases to show $\reg(\alg', v) \leq 1/4$.
\end{itemize}
Overall, we have $\ben \leq \reg(\alg') \leq 1/4 \leq \reg(\alg) / 2$, for any deterministic $\alg$.
\end{proof}

\monotone*
\begin{proof}
Let $\alpha$ be $\frac{1}{\max\left( 2, \lceil \frac{6}{ \varepsilon} \rceil\right)}$. Therefore $\alpha \leq \min(\varepsilon /6, 1/2)$ and $1/\alpha$ is an integer. Define $A_1$ to be uniformly distributed on $[-1, -1 + \alpha] \cup [ 1 - \alpha, 1]$. 

For any deterministic and monotone algorithm $\alg$, it picks either Item 1 or 2 when observe $s_1 = 0$. Without loss of generality we assume $\alg$ picks Item 1 when observe $s_1 = 0$. Since $\alg$ is monotone, we know $\alg$ picks Item 1, when observe $s_1 \geq 0$. We set $v_1 = -1 + \alpha$. $\alg$ observes $s_1 \geq 0$ and pick Item 1 with probability $1/2$. So $\reg(\alg) \geq (1-\alpha) / 2$.

Now consider the following deterministic (but not monotone) algorithm $\alg'$. For each observed value $s_1$, decompose $s_1$ as $s_1 = p \cdot \alpha + q$ where $p = \lfloor \frac{s_1}{\alpha} \rfloor$. And $\alg'$ picks Item 1 if $q / \alpha \leq (p\cdot \alpha + 1)/2$. The following \cref{fig:sep_mono} illustrates how $\alg'$ picks an item based on $s_1$. Intuitively, $\alg'$ is mimicking the randomized algorithm in \cref{thm:sep_random}. Even as a deterministic algorithm, it can do so because the noise distribution $A_1$ is distributed continuously in two ranges. Notice that $\alg'$ switches $O(1/\alpha)$ times between picking Item 1 and Item 2 when $s_1$ goes from -2 to 2, and this cannot be done by a monotone and deterministic algorithm.
It is easy to check that if $s_1$ distributed uniformly in $[k, k+\alpha]$ for some $k$, the probability of $\alg'$ picks Item 1 is in $\left[\min((k - \alpha + 1)/2, 1),\min((k + \alpha + 1)/2, 1) \right]$.

\begin{figure}
\label{fig:sep_mono}
\centering
\begin{tikzpicture}
    \draw[blue,line width=1mm] (-7.5,0) -- (7.5,0);
    \draw[red,line width=1mm] (-4.5,0) -- (-4.5+1.5/8,0);
    \draw[red,line width=1mm] (-3,0) -- (-3+1.5/4,0);
    \draw[red,line width=1mm] (-1.5,0) -- (-1.5 + 1.5/8*3,0);
    \draw[red,line width=1mm] (0,0) -- (0+1.5/8*4,0);
    \draw[red,line width=1mm] (1.5,0) -- (1.5+1.5/8*5,0);
    \draw[red,line width=1mm] (3,0) -- (3+1.5/8*6,0);
    \draw[red,line width=1mm] (4.5,0) -- (4.5+1.5/8*7,0);
    \draw[red,line width=1mm] (6,0) -- (7.5,0);
    \draw[white, draw opacity =0,] (-7.5,0) -- (7.5,0)
    node[pos=0.1,style={fill,circle,inner sep=2pt,outer sep=0pt},fill=black,text=blue,label=above:\textcolor{black}{$-4\alpha$}]{}
    node[pos=0.2,style={fill,circle,inner sep=2pt,outer sep=0pt},fill=black,text=blue,label=above:\textcolor{black}{$-3\alpha$}]{}
    node[pos=0.3,style={fill,circle,inner sep=2pt,outer sep=0pt},fill=black,text=blue,label=above:\textcolor{black}{$-2\alpha$}]{}
    node[pos=0.4,style={fill,circle,inner sep=2pt,outer sep=0pt},fill=black,text=blue,label=above:\textcolor{black}{$-\alpha$}]{}
    node[pos=0.5,style={fill,circle,inner sep=2pt,outer sep=0pt},fill=black,text=blue,label=above:\textcolor{black}{0}]{}
    node[pos=0.6,style={fill,circle,inner sep=2pt,outer sep=0pt},fill=black,text=blue,label=above:\textcolor{black}{$\alpha$}]{}
    node[pos=0.7,style={fill,circle,inner sep=2pt,outer sep=0pt},fill=black,text=blue,label=above:\textcolor{black}{$2\alpha$}]{}
    node[pos=0.8,style={fill,circle,inner sep=2pt,outer sep=0pt},fill=black,text=blue,label=above:\textcolor{black}{$3\alpha$}]{}
    node[pos=0.9,style={fill,circle,inner sep=2pt,outer sep=0pt},fill=black,text=blue,label=above:\textcolor{black}{$4\alpha$}]{};
    \draw[red,line width=1mm] (4.5,2) -- (5.0,2)
    node[pos=1,label=right:\textcolor{black}{Pick Item 1}]{};
    \draw[blue,line width=1mm] (4.5,1.2) -- (5.0,1.2) node[pos=1,label=right:\textcolor{black}{Pick Item 2}]{};
\end{tikzpicture}
\caption{$\alg'$ for $\alpha = 1/2$}
\end{figure}
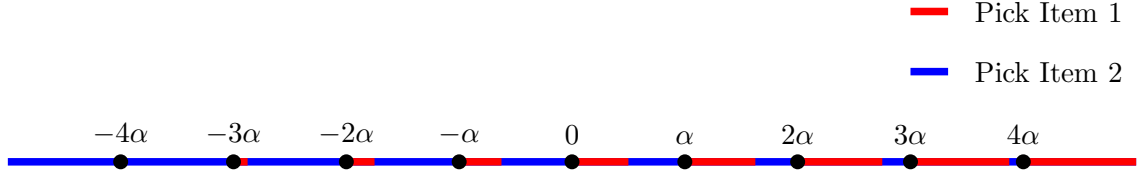
  
Now we bound $\reg(\alg', v)$ for different cases of $v$'s:
\begin{itemize}
    \item For any $v_1 \geq 2 + \alpha$, we know $s_1 \geq 1 + \alpha$, and then $\alg'$ picks Item 1 with probability $1$. Therefore, in this case, $\reg(\alg', v) = 0$.
    \item For any $v_1 \in (\alpha, 2+\alpha)$, we know $s_1 > 1 + \alpha$ only when $A_1 = -1$. $\reg(\alg', v) =\frac{1}{2} \cdot  v_1 \cdot (1 - (v_1 - \alpha)/2) \leq 1/4 + \alpha / 2$.
    \item For any $v_1 \in [0, \alpha]$, $\reg(\alg', v) \leq \alpha \leq 1/4 + \alpha / 2$.
    \item For any $v_1 < 0$, we can use a similar argument as the previous two cases to show $\reg(\alg', v) \leq 1/4 + \alpha / 2$.
\end{itemize}
Overall, we have
\[
\reg(\alg') \leq 1/4 + \alpha/2 < \frac{1/2 - \alpha/ 2}{2 - 6\alpha} \leq \frac{\reg(\alg)}{2-\varepsilon}. \qedhere
\]
\end{proof}

\section{Conclusions}
\label{sec:conclusions}
In this paper, we formulated a problem of minimizing regret when selecting one of many options in a noisy environment. We identified a simple algorithm for this task that satisfies nice properties (deterministic, monotone, and efficiently computable) and provides constant-approximation guarantees. 

We believe this is a promising research direction, ripe with problems that are both mathematically challenging and elegant and practically well-motivated. Here we name a few:

\begin{itemize}
\item An immediate future direction is to generalize our framework to more complex settings than selecting one from $n$ items. For example, it is an interesting question to find an approximately shortest path when the edge lengths can only be observed through noisy channels. This is motivated, for example, by applications in map navigation in presence of traffic: navigation services might observe the amount of time that it takes their users to traverse the roads they are traveling on, and they also have knowledge about the noise in these estimates (for example, having a traffic light on a road adds a random delay, with a known distribution). Using this information, they might want to propose a route to a user looking to go from one point on the map to another.

We point out that simply applying our $\off_\theta$ algorithm to all edges does not provide a constant-approximation in regret in the shortest-path setting, due to potential accumulation of error. We leave this open direction to future research.

\item In several places in this paper (the threshold formula $\theta(D)$, the example in the proof of \cref{thm:expectation}) we encountered similarities with the pricing literature (e.g., the monopoly pricing formula, the equal revenue distribution). Is there any deeper connection between the two settings?

\item Can we prove a better constant bound on the performance of the $\off_\theta$ algorithm? The worst example we know is a factor of 2 (which follows from the results in \cref{sec:lower}).

\item Is there a simpler way to prove \cref{cor:symm} that does not go through the analysis of $\off_\theta$ (and perhaps achieves a better bound)?

\item While the examples in \cref{sec:lower} illustrated that the optimal algorithm can have a complex form, we are not aware of any computational hardness result for this problem. Is such a result possible?

\end{itemize}

\bibliography{bib}
\bibliographystyle{alpha}

\end{document}